\documentclass[a4paper]{IEEEtran}
\usepackage{amsfonts,amsmath,amsthm,amssymb,url}
\usepackage[hidelinks,bookmarks=false]{hyperref}
\hypersetup{
  pdfauthor={L.Csirmaz},
  pdftitle={Revisiting the Stability of the Ingleton Inequality: A
            Tropicalization-Free Approach},
  bookmarksopen=false,
  pdfkeywords={Shannon entropy; Ingleton inequality; entropy region; information inequality; polymatroid; Ahlswede--Korner reduction},
}

\let\>\rangle\let\<\langle
\let\norm\|
\def\|{\mskip1.5mu{|}\mskip1.5mu}
\def\H{\mathbf H}
\def\I{\mathbf I}
\def\m#1{\mskip 1.5mu{#1}\mskip 1.5mu}
\newcommand\restr{\mkern1mu{\upharpoonright}\mkern1mu}

\def\<{\langle}\def\>{\rangle}
\newcommand\GN{\Gamma_{\!N}}
\def\Ga#1{\Gamma^*_{\!#1}}
\def\barGamma{\mkern1mu\overline{\hbox{$\mkern-1mu\Gamma\mkern4mu$}}\mkern-4mu}
\def\clGa#1{\barGamma^*_{\!#1}}
\DeclareMathOperator\Prob{\mathrm{Prob}}
\newcommand\eqdef{\stackrel{\mathrm{def}}{=}}
\newcommand\dn{{\downarrow}}

\def\ing[#1,#2,#3,#4]{\mathsf{ing}(#1,#2,#3,#4\mkern 0mu)}

\newtheorem{theorem}{Theorem}
\newtheorem{lemma}[theorem]{Lemma}

\newtheorem{claim}[theorem]{Claim}
\newtheorem{construction}{Construction}
\theoremstyle{definition}

\newenvironment{Keywords}[1]{\xdef\IEEEkeywordsname{#1}\IEEEkeywords}{\endIEEEkeywords}
\title{\bfseries\fontsize{17.4}{19}\selectfont
Revisiting the Stability of the Ingleton Inequality: A Tropicalization-Free Approach
}
\author{\fontsize{12.4}{12}\selectfont
Laszlo Csirmaz\IEEEauthorrefmark1
\thanks{\IEEEauthorrefmark1e-mail:csirmaz@renyi.hu\\
\hbox to 10pt{\hfil} R\'enyi Institute, Hungary; UTIA, Prague, Czech Republic}}

\begin{document}

\maketitle

\begin{abstract}

\noindent
The classical Ingleton inequality is known to hold for entropic points under
specific exact conditional independence constraints. Recently, Matveev and
Romashchenko (2026) investigated the stability of these implications,
quantifying the extent to which the Ingleton inequality can be violated when
a group of conditional mutual information terms is small but non-zero. While
their proofs relied fundamentally on the complex framework of tropical
probability spaces, we revisit these stability results using a completely
tropicalization-free approach. By developing an alternative framework, we
significantly streamline the underlying concepts and proofs, derive explicit
error terms, and improve some estimates. Furthermore, we resolve an open
problem posed in prior work by exhibiting a new infinite family of entropy
inequalities that establishes the stability of the sum of two Ingleton
expressions.


\end{abstract}

\section{Introduction}\label{sec:intro}

Ingleton's inequality plays an important role in the geometry of the
four-variable entropy region. It was originally stated in \cite{ingleton} as a
necessary condition for a matroid to be representable \cite{oxley}. The
smallest matroid violating this inequality---and also the smallest
non-representable one---is the V\'amos matroid on eight elements
\cite{oxley}. In the late 1990s, through a series of works by F.~Mat\'u\v{s} and
M.~Studen\'y \cite{Ma.Stud, mat95, mat99}, the Ingleton inequality emerged as a
crucial parameter in the geometric description of the four-variable entropy
region, $\Ga4$. For a more recent and comprehensive overview of these results,
consult \cite{MP20} and \cite{studeny21}. 

While the global structure of $\Ga4$ remains far from being fully understood 
despite considerable efforts devoted to this problem, there are significant 
achievements regarding the local behavior of the boundary of $\Ga4$ where the 
Ingleton inequality starts to fail. Results in a recent comprehensive work by 
Matveev and Romashchenko \cite{rom3} quantify the extent to which the Ingleton 
inequality is violated around these boundary points. These theorems provide 
information about how smoothly the curved part of the boundary approaches a 
given flat face. 

The work in \cite{rom3} relies fundamentally on the complex framework of
\emph{tropical probability spaces}, introduced in \cite{MP18} and further
developed in \cite{MP20a}. Prior to \cite{rom3}, this framework was applied
successfully to investigate the global structure of $\Ga4$ in \cite{MP20},
among other applications. Tropicalization is an excellent tool for investigating the
asymptotic behavior of discrete random variables, as it preserves many
subtle properties of probability spaces that are otherwise lost when working
with traditional entropy profiles.

The main objective of this work is to develop an alternative, traditional
framework. This approach, on one hand, avoids tropicalization, and, on the 
other, allows us to prove both the technical lemmas (stated for almost 
entropic vectors instead of tropical tuples) and all main theorems from \cite{rom3}. 
In addition to achieving this goal and streamlining the underlying concepts, 
our further contributions are as follows:
\begin{itemize}
\item All error terms are presented explicitly, rather than relying on asymptotic
 $O(\boldsymbol\cdot)$ notation.
\item Certain conditions are relaxed in order to strengthen the estimates.
\item An open problem posed in \cite{rom3} is resolved by exhibiting a new
infinite family of entropy inequalities.
\end{itemize}

The rest of the paper is organized as follows. Section \ref{sec:defs}
recalls the basic definitions and notations of Information Theory, including
the Copy Lemma \cite{DFZ11}. Section \ref{sec:construction} describes the
main construction and proves all properties used later. Section \ref{sec:AK}
introduces the so-called Ahlswede-K\"orner reduction \cite{AK, MMRV,
M.fmadhe}, and applies it to the main construction. Theorems estimating the
stability of the Ingleton inequality under different conditions are
presented in Section \ref{sec:stab}. Finally, Section \ref{sec:conc} concludes
the paper.


\section{Notation}\label{sec:defs}

In this paper, all sets are finite. Sets are denoted by capital letters,
such as $A$, $B$, $N$, or $T$, and elements of these sets by lowercase
letters. The union sign and curly brackets around singletons are often
omitted; thus, $ABi$ denotes the set $A\cup B\cup\{i\}$. Unless indicated
otherwise, the notation $N=ABa$ implicitly indicates that the sets on the
right-hand size are disjoint, forming a \emph{partition} or a decomposition
of the left-hand side set $N$.

The \emph{entropy profile} of jointly distributed discrete random variables
$\xi =\<\xi_i:i\in N\>$ indexed by elements of the finite \emph{base set
$N$} is the function on the non-empty subsets of $N$ whose value evaluated
on $J\subseteq N$ is the standard Shannon entropy $\H(\xi_J)$ of the
marginal distribution $\<\xi_i:i\in J\>$ (see, e.g., \cite{yeung-book}). The
collection of all entropy profiles is denoted by $\Ga N$. Elements of $\Ga
N$ are considered interchangeably as vectors indexed by subsets of $N$, as
points in the $(2^{|N|}\m-1)$-dimensional Euclidean space, and as
functions assigning non-negative real numbers to non-empty subsets of $N$.
In particular, the entropy profile of four jointly distributed random
variables is a $15$-dimensional real vector. Points (or functions) in $\Ga
N$ are called \emph{entropic}, and points (or functions) in its closure
$\clGa N$ (in the usual Euclidean sense) are \emph{almost entropic}, or
\emph{aent}. The set $\clGa N$ is a full-dimensional convex cone
\cite{yeung-book}, while internal points of $\clGa N$ are entropic
\cite{fmtwocon}; thus, $\clGa N\setminus \Ga N$ contains only boundary
points.

For a joint distribution $\xi$ over the base set $N$, the entropy
$\H(\xi_J)$ of the marginal $\xi_J$ is also denoted by $\H_\xi(J)$; the
index $\xi$ is omitted, and we write $\H(J)$, when the distribution $\xi$ is
clear from the context. The traditional information-theoretic notation is
formally extended to an arbitrary function $f$ defined on subsets of $N$ as
$$\begin{array}{l@{\:\eqdef\:}l@{}l}
  \H_f(J) & f(J), \\[2pt]
  \H_f(J\|K) & f(J\|K) &{\:\eqdef\:} f(JK)-f(K), \\[2pt]
  \I_f(J,K)  & f(J,K)  &{\:\eqdef\:} f(J)+f(K)-f(JK), \\[2pt]
  \I_f(J,K\|L) & f(J,K\|L)&{\:\eqdef\:}\\
               & \hbox to 0pt{$f(JL)+f(KL)-f(L)-f(JKL).$\hss}
\end{array}$$
Similarly to distributions, the function $f$ is omitted when it is clear and
unambiguous from the context. The non-negativity of these expressions for
entropic (and, by continuity, for almost entropic) points constitutes the
\emph{Shannon inequalities}. A function that respects the non-negativity of
these expressions is called a (rank function of a) \emph{polymatroid}. The
collection of polymatroids on (the subsets of) $N$ is denoted by $\GN$.
$\GN$ is a polyhedral cone and clearly contains $\clGa N$. A linear
inequality valid for all points of $\GN$ is called a \emph{Shannon
inequality}. Due to the strong duality of linear programming, or to the
Farkas Lemma \cite{ziegler}, checking whether an inequality is Shannon is
equivalent to checking whether an LP (linear program) instance has a
feasible solution. Automated inequality provers, such as \textsf{oXitip}
\cite{ITIP}, \textsf{Psitip} \cite{Li23}, and many others \cite{W-auto},
work seamlessly for $|N|\le 14$. For this reason, we do not
indicate why an inequality is Shannon, but rather state it as a fact.
 
The notation $(f,N)$ emphasizes that the function $f$ is defined on the
non-empty subsets of the base set $N$. By a slight abuse of notation,
$f\restr M$ for $M\subseteq N$ denotes the restriction of $f$ to the
non-empty subsets of $M$. We remark that if $(f,N)$ is a polymatroid (or
entropic, or almost entropic), then so is $f\restr M$.

The elements $a,b\in N$ are \emph{symmetric} with respect to $X\subseteq N$
and $(f,N)$, written as $aX\equiv^f bX$, if $f(aJ)=f(bJ)$ for all $J
\subseteq X\setminus\{a,b\}$; and they are \emph{symmetric}, or
\emph{$f$-symmetric}, if $aN\equiv^f bN$. We remark that if $(f,N)$ is
entropic and $a,b\in N$ are $f$-symmetric, then the underlying distribution
need not be symmetric (and cannot be assumed to be so), while the converse
implication evidently holds. In contrast, in this case theres always exists a
symmetric tropical probability distribution whose entropy profile is $f$.

\subsection{Copy Lemma}\label{subsec:copy}

Suppose $X$, $M$ and $Y$ are disjoint sets, $\xi$ is a distribution on $XM$,
and $\eta$ is a distribution on $MY$ such that the marginals $\xi_M$ and
$\eta_M$ are identical. In other words, if $\xi$ is defined on the alphabet
$\mathcal X\times\mathcal M$, and $\eta$ is defined on the alphabet
$\mathcal M\times\mathcal Y$, then, for all $\beta\in\mathcal M$, the
probability that $\xi_M$ takes the value $\beta$ is the same as the
probability that $\eta_M$ takes the value $\beta$:
$$
   \Prob(\xi_M\m=\beta)=\Prob(\eta_M\m=\beta) ~~\mbox{ for every }
   \beta\in\mathcal M.
$$
In this case, one can glue these distributions along their common marginals,
and this can be done such that $X$ and $Y$ become independent given $M$. This
simple observation is behind the most powerful tool for investigating the
entropy region. We state two lemmas that are immediate from this observation.
The first one, called the \emph{Copy Lemma}, appeared implicitly in the
original work of Zhang and Yeung \cite{ZhY.conic, ZhY.ineq}, and was later
formulated explicitly in \cite{DFZ11}. In fact, the Copy Lemma is an
immediate consequence of Lemma \ref{lemma:maxe}, since duplicating a base
element clearly preserves entropic and almost entropic points. The
construction described in the Copy Lemma is referred to as adding $r$ as a copy
of $a$ over $X$.

\begin{lemma}[Copy Lemma]\label{lemma:copy}
Suppose $N$ is partitioned as $N=aXT$, and $(f,N)$ is almost entropic.
There exists an aent function $(g,rN)$ such that
\begin{itemize}
\item[a)] $g$ is an aent extension of $f$: $g\restr N=f$;
\item[b)] $a$ and $r$ are symmetric with respect to $X$: $aX\equiv^g rX$;
\item[c)] $r$ and $aT$ are conditionally independent
given $X$:

$\I_g(r,aT\|X)=0$.
\end{itemize}
Additionally, if neither $u$ nor $v$ is in $X$ and they are $f$-symmetric,
then $g$ can be chosen such that $u$ and $v$ remain $g$-symmetric as well.
\end{lemma}

\begin{proof}
It suffices to prove the lemma for entropic functions only. Indeed, if $f_i\to
f$, where $f_i$ is entropic and $g_i$ is a corresponding extension, then
the sequence $g_i$ is bounded. Consequently, it has a convergent subsequence
whose limit satisfies the conditions for $g$.

If $(f,N)$ is the entropy profile of the distribution $\xi=\<\xi_i:i\in
N\>$, then add an identical copy $r$ of $a$ to this distribution, and apply
the observation to the distributions $\xi_{rX}$ and $\xi_{XaT}$.

To prove the last claim, let $\pi$ be the permutation that swaps $u$
and $v$ and keeps every other element fixed. Define $\pi f$ by $(\pi
f)(A)=f(\pi A)$. The assumption that $u$ and $v$ are
$f$-symmetric can be expressed equivalently as $f=\pi f$. If this is the
case and $g$ satisfies conditions a)--c), then so does $\pi g$. Finally,
let $g^*=(g+\pi g)/2$, the average of $g$ and $\pi g$. Then $g^*=\pi g^*$,
and $g^*$ also satisfies conditions a)--c), as required.
\end{proof}

\begin{lemma}\label{lemma:maxe}
Suppose $N$ is partitioned as $N=XMY$, and $(f,N)$ is almost entropic.
There exists an aent function $(g,N)$ such that
\begin{itemize}
\item[a)] $g\restr XM = f\restr XM$, and $g\restr MY=f\restr MY$; and
\item[b)] $\I_g(X,Y\|M)=0$.
\end{itemize}
\end{lemma}
\begin{proof}
If $f$ is the entropy profile of $\xi=\<\xi_i:i\in N\>$, then apply the
observation to the distributions $\xi_{XM}$ and $\xi_{MY}$. For aent functions,
use a convergent subsequence as in the proof of Lemma \ref{lemma:copy}.
\end{proof}


\section{The construction}\label{sec:construction}

\begin{construction}\label{defz}
Partition the base set as $N=aMT$. For every aent function $(f,N)$ and $k\ge 0$,
there exists an aent extension $(f,z_0z_1\dots z_kN)$ such that
\begin{itemize}
\item[a)] $z_0$ is trivial: it takes a single value with probability $1$;
\item[b)] $\H(Saz_i) = \H(Sz_{i+1})$ for every $0\le i<k$ and every $S\subseteq M$;
\item[c)] $\H(z_i\|JM)=i\cdot\H(a\|M)$ for every $0\le i\le k$ and $J\subseteq aT$.
\end{itemize}
Additionally, if $a$ and $b\in T$ are $f$-symmetric, then there
is such an extension in which $a$ and $b$ remain $f$-symmetric.
\end{construction}

\begin{proof}
By induction on $k$. For $k=0$ the statement is trivial. Suppose we have
such an extension for $k$. Apply the Copy Lemma \ref{lemma:copy} to the
partition
$$
     a\cup Mz_k \cup Tz_0\dots z_{k-1}
$$
of the base set to get an extension in which $r$ is a copy of $a$ over
$Mz_k$. Then $aMz_k \equiv rMz_k$ and $\I(r,aT\|Mz_k)=0$. The new variable
$z_{k+1}$ will be the pair $rz_k$. The symmetry $aMz_k \equiv rMz_k$
immediately gives $\H(Saz_k)= \H(Srz_k)= \H(Sz_{k+1})$ for all $S\subseteq
M$, proving b).

If $J\subseteq aT$, then we have $\I(r,J\|Mz_k)=0$. Using the fact that $\H(rMz_k)=
\H(aMz_k)$ by b), this conditional independence rewrites to
\begin{align*}
  & \H(z_{k+1}\|JM) = {}\\
  & ~~~ \H(z_k\|JM) + \H(z_k\|aM) - \H(z_k\|M) + \H(a\|M).
\end{align*}
By induction, the right-hand side equals $(k+1)\cdot\H(a\|M)$, proving c).

Finally, the statement about $f$-symmetry is immediate from the $f$-symmetry 
stated in the Copy Lemma.
\end{proof}

\begin{claim}\label{claim:1}
For $S\subseteq M$, we have $\H(z_i\|S)\le i\cdot\H(a\|S)$.
\end{claim}
\begin{proof}
By submodularity, 
$$\H(z_{i+1}\|S)=\H(az_i\|S)\le \H(z_i\|S)+\H(a\|S).$$
The claim follows by induction on $i$.
\end{proof}

\begin{lemma}\label{lemma:D}
The following statements hold for every partition $(\bar S, S)$ of $M$:
\begin{itemize}
\item[a)] $\H(z_i\|M)=i\cdot\H(a\|M)$;
\item[b)] $\I(\bar S,S\|z_{i+1})=\I(\bar S,S\|az_i)$;
\item[c)] $\H(z_i\|S) -\H(z_i\|M) \le i\cdot\I(\bar S,a\|S)$.
\end{itemize}
\end{lemma}
\begin{proof}
Property a) is immediate from claim c) of the Construction (by setting $J=\emptyset$), while property
b) follows from claim b) there. Finally, for property c), observe that
$$
   \I(\bar S,a\|S)=\H(a\|S)-\H(a\|M).
$$
From Claim \ref{claim:1}, we have $\H(z_i\|S)\le i\cdot\H(a\|S)$;
property a) gives $\H(z_i\|M)=i\cdot\H(a\|M)$; these together yield the
required inequality.
\end{proof}


\section{Ahlswede--K\"orner reduction}\label{sec:AK}

The main technical ingredient of the Ahlswede--K\"orner reduction is a
piecewise linear map that preserves almost entropic points. This statement,
presented without proof as Lemma \ref{lemma:down} below, is known as
the \emph{Ahlswede--K\"orner lemma}. It is implicit in \cite{AK}
and stated explicitly in \cite{Kaced}. Its proof is quite involved; it can
be found, e.g., in \cite[Lemma 5]{MMRV}. A more general form of this lemma
appears as Lemma 3 in \cite{entreg}, but see also Theorem 3 in \cite{fmtwocon}.

\begin{lemma}[Ahlswede--K\"orner lemma]\label{lemma:down}
Let $(f,N)$ be almost entropic, $w\in N$, and set $\alpha=f(N)-f(N\setminus w)$.
The function $f\dn^w$ defined on the subsets of the base set $N$ as
$$
    f\dn^w (A) = \begin{cases}
       f(A) & \mbox{ if $w\notin A$}, \\
       f(A)-\alpha & \mbox{ if $w\in A$},
    \end{cases}
$$
is also almost entropic.\qed
\end{lemma}

The reduction $\dn^w$ also maps polymatroids to polymatroids. This means
that in order to separate aent and non-aent points, Lemma \ref{lemma:down}
must be combined with some additional machinery.

\begin{lemma}[Ahlswede--K\"orner reduction]\label{lemma:AK}
Let $N=MXz$ be a partition, and let $(f,N)$ be aent. There exists an 
aent extension $(g,Nw)$ of $f$ such that
\begin{itemize}
\item[a)] $g\restr N = f$,
\item[b)] $\H_g(J\|w)=\H_f(J\|z)$ whenever $J\subseteq M$, and
\item[c)] $\H_g(w\|M)=0$.
\end{itemize}
Additionally, if $a,b\in X$ are $f$-symmetric, then $g$ can be chosen 
such that they remain $g$-symmetric as well.
\end{lemma}

\begin{proof}
First, extend $(f,N)$ by $w$ as a copy of $z$ over $M$. According to the
Copy Lemma, we have $wM\equiv^f zM$ and $\I_f(w,zX\|M)=0$; in
particular, $\H_f(wJ)=\H_f(zJ) \mbox{ for all } J\subseteq M$. Apply Lemma
\ref{lemma:down} to the aent function $(f,wMXz)$ and element $w$. The
constant is $\alpha=\H_f(w\|MXz)=\H_f(w\|M)$. Denoting $f\dn^w$ by $g$,
we have
$$
   g(A) = \begin{cases}
      f(A) & \mbox{ if $w\notin A$}, \\
      f(A)-\alpha & \mbox{ if $w\in A$}.
   \end{cases}
$$
Therefore, $g$ is an extension of $f$. If $J\subseteq M$, then
$g(wJ)=f(wJ)-\alpha$; thus,
\begin{align*}
   \H_g(J\|w)&=g(wJ)-g(w) = f(wJ)-f(w)={}\\
             &=f(zJ)-f(z)=\H_f(J\|z).
\end{align*}
Finally, $g(wM)=f(wM)-\alpha=f(M)=g(M)$; thus, $\H_g(w\|M)=0$.
If $a$ and $b$ are $f$-symmetric, then use the symmetry-preserving version
of the Copy Lemma. Lemma \ref{lemma:down} clearly preserves this symmetry.
\end{proof}

\begin{construction}\label{lemma:Dw}
Partition the base set as $N=aMT$. For every aent function $(f,N)$ and $k\ge
0$, there exists an aent extension $(f,w_0w_1\dots w_kN)$ such that $w_0$ is
trivial (it takes a single value with probability $1$); moreover, the
following holds for every partition $(S,\bar S)$ of $M$:
\begin{itemize}
\item[a)] $\H(w_i\|M)=0$,
\item[b)] $\I(S,\bar S\|w_{i+1}) = \I(S,\bar S\|aw_i)$,
\item[c)] $\H(w_i\|S)\le i\cdot\I(\bar S,a\|S)$.
\end{itemize}
Additionally, if $a$ and $b\in T$ are $f$-symmetric, then there exists an
extension in which they remain $f$-symmetric.
\end{construction}

\begin{proof}
Use Construction \ref{defz} to create the aent extension $(f,z_0\dots z_k N)$; 
then apply the Ahlswede--K\"orner reduction, Lemma \ref{lemma:AK}, 
successively for $i=0,1,\dots,k$ using the partition
$$
   aMT\, \cup \, \big(\{z_j:j\neq i\}\cup\{w_j:j<i\}\big)\, \cup\, \{z_i\}
$$
to add the variable $w_i$. Lemma \ref{lemma:AK} yields
$\H(J\|w_i)=\H(J\|z_i)$ for all $J\subseteq aMT$ and $\H(w_i\|aMT)=0$.
Applying this to $J=aMT$, we get
\begin{align*}
   \H(z_i)-\H(w_i)&=\H(aMTz_i)-\H(aMTw_i) ={}\\
                  &=\H(z_i\|aMT)=\H(z_i\|M)
\end{align*}
by point c) of Construction \ref{defz}. Thus,
$$
  \H(w_iJ)=\H(z_iJ)-\H(z_i\|M)
$$
whenever $J\subseteq aMT$. Plugging in these values, a) is immediate; b) and
c) follow from the corresponding points in Lemma \ref{lemma:D}.

The symmetry statement is immediate from Construction \ref{defz} and the
Ahl\-swede--K\"orner reduction.
\end{proof}


\section{Stability of the Ingleton inequality}\label{sec:stab}

We have developed the main tools that will be used to investigate the
stability of the Ingleton inequality. The first tool, Construction
\ref{lemma:Dw} from Section \ref{sec:AK}, guarantees, for any aent function,
an aent extension with specific properties. The second one, Lemma
\ref{lemma:maxe} from Section \ref{sec:defs}, states that, under certain
conditions, the validity of a specific conditional independence can be
assumed.

Ingleton's inequality plays an important role in the geometry of the
four-variable entropy region. The Ingleton expression for four variables is
defined as
$$
   \ing[x,y,a,b] \eqdef -\I(x,y)+\I(x,y\|a)+\I(x,y\|b)+\I(a,b),
$$
and has several equivalent forms (see \cite{studeny21}). The inequality
$\ing[x,y,a,b]\ge 0$ holds for representable matroids \cite{oxley}, and was
introduced by A.~W.~Ingleton specifically for this purpose \cite{ingleton}. In
the late 1990s, through a series of works by F.~Mat\'u\v{s} and M.~Studen\'y
\cite{Ma.Stud, mat95, mat99}, the Ingleton inequality emerged as a crucial
parameter in the geometric description of the four-variable entropy region
$\Ga4$. For four variables, there are six non-equivalent instances of this
inequality, and no two of them can be violated simultaneously. The region
where all six Ingleton inequalities hold forms the \emph{Ingleton part} of
$\Ga4$. This is a polyhedral cone with an explicit enumeration of all facets
and extremal rays. The remaining part splits into six disjoint, congruent
cones, depending on which Ingleton instance is violated. Each of them is
contained in a $15$-dimensional simplicial cone: the base is the hyperplane
where the corresponding Ingleton expression is zero, surrounded by $14$
hyperplanes where certain conditional independences (or functional
dependencies) hold. Investigating the stability of the Ingleton inequality
means estimating the extent to which it can be violated in the proximity of
those boundary hyperplanes. For example, Theorem \ref{thm:AB} gives a lower
bound on the Ingleton expression assuming that both $\I(x,a\|y)$ and
$\I(y,a\|x)$ are small in terms of $\I(x,y)$. This estimate provides
information on the geometry of $\Ga4$ at the intersection of the Ingleton
base and the hyperplanes corresponding to these conditional independences.
Theorem \ref{thm:C} provides a stronger (almost linear) estimate when,
additionally, the symmetrical $\I(x,b\|y)$ and $\I(y,b\|x)$ expressions are
also small. This estimate suggests a different geometrical structure at these
``corner points.''

In the rest of this section, the base set $N$ is fixed to $xyab$ or
$xyabcd$, and the subset $M\subseteq N$ is fixed to $\{x,y\}$. We also fix the
almost entropic function $(f,N)$, and take the aent extension
$(f,w_0w_1\dots w_kN)$ guaranteed by Construction \ref{lemma:Dw} from
Section \ref{sec:AK} for some integer parameter $k\ge 1$ to be determined
later. We also assume the extension to be $(a,b)$-symmetric when the
original aent function was $(a,b)$-symmetric.

First, we state and prove the analogue of Lemma E in \cite{rom3}, but
without tropicalization and with explicit error terms. Theorems
\ref{thm:AB}, \ref{thm:F} and \ref{thm:C} provide estimates under different
conditions, with explicit error terms, on the possible violations of the
Ingleton inequality.

\begin{lemma}\label{corr:E}
Suppose the integer $k>1$ is such that
$$
  \I(x,a\|y)+\I(y,a\|x)\le (1/k^2)\,\I(x,y).
$$
Then for some $i<k$ we have
\begin{itemize}
\item[a)] $\H(w_i\|x)+\H(w_i\|y) \le (1/k)\,\I(x,y)$, and
\item[b)] $\I(xy,a\|w_i)\le (2/k)\,\I(x,y)$.
\end{itemize}
\end{lemma}
\begin{proof}
a) is immediate from point c) of Construction \ref{lemma:Dw}. To prove b), we start
from the four-variable Shannon inequality
\begin{align*}
  \I(xy,a\|w) &\le \I(x,a\|y)+\I(y,a\|x)+\H(w\|x)+\H(w\|y) + {}\\
              & ~~~~~~{} + \I(x,y\|w)-\I(x,y\|wa).
\end{align*}
According to Construction \ref{lemma:Dw}, this implies
$$
   \I(xy,a\|w_i)\le (i+1)\big(\I(x,a\|y)+\I(y,a\|x)\big ) +
   \delta_i,
$$
where $\delta_i=\I(x,y\|w_i)-\I(x,y\|w_{i+1})$. Since
$$
    \delta_0+\delta_1+\cdots+\delta_{k-1}\le 
    \I(x,y\|w_0)=\I(x,y),
$$
there is an $i< k$ such that $\delta_i\le (1/k)\I(x,y)$. Then $i+1\le
k$; thus,
\begin{align*}
  \I(xy,a\|w_i) &\le k\big(\I(x,a\|y)+\I(y,a\|x)\big ) +
        \delta_i \le{} \\
   &\le (1/k)\I(x,y) + (1/k)\I(x,y),
\end{align*}
proving b).
\end{proof}

\begin{theorem}\label{thm:AB}
Suppose $k>1$ is an integer such that
$$
   \I(x,a\|y)+\I(y,a\|x) \le (1/k^2)\I(x,y).
$$
Then both $\ing[x,y,a,b]$ and $\ing[x,a,y,b]$ are at least

$-(3/k)\I(x,y)$.
\end{theorem}
\begin{proof}
These are Shannon inequalities:
\begin{align*}
   \ing[x,y,a,b] &\ge -\I(xy,a\|w)-\H(w\|x)-\H(w\|y), \\
   \ing[x,a,y,b] &\ge -\I(xy,a\|w)-\H(w\|x)-\H(w\|y).
\end{align*}
Apply Lemma \ref{corr:E}.
\end{proof}

\begin{theorem}\label{thm:F}
Suppose $k>1$ is an integer such that
\begin{align*}
  \I(x,a\|y)+\I(y,a\|x) &\le (1/k^2)\,\I(x,y), \mbox{ and }\\
  \I(x,b\|y)+\I(y,b\|x) &\le (1/k)\,\I(x,y).
\end{align*}
Then
\begin{equation}\label{eq:F}
   \ing[x,y,a,b]+\ing[x,y,c,d]+ (5/k)\,\I(x,y) \ge 0.
\end{equation}
\end{theorem}

An immediate consequence is that under the conditions of the theorem, the
following alternative holds for any real number $t$:
$$\begin{array}{rl}
   \mbox{\emph{either}} & \ing[x,y,a,b] + (3/k-t)\I(x,y)\ge 0,\\[2pt]
   \mbox{\emph{or} }    & \ing[x,y,c,d] + (2/k+t)\I(x,y)\ge 0.
\end{array}$$
Note that the assumptions do not involve the variables $c$ and $d$, so the
conclusion holds for any $c$ and $d$ jointly distributed with $xyab$.

\begin{proof}
It suffices to prove (\ref{eq:F}) only, as then the two statements
in the last claim cannot fail simultaneously. We use the following Shannon
inequality:
\begin{align*}
  & \ing[x,y,a,b] + \ing[x,y,c,d] \ge -3\,\I(ab,cd\|xy) - {} \\
  & \mkern 40mu {}-\I(xy,a\|w)-2\big(\H(w\|x)+\H(w\|y)\big) - {} \\
  & \mkern 40mu {}-\big(\I(x,b\|y)+\I(y,b\|x)\big).
\end{align*}
Using Lemma \ref{lemma:maxe} from Section \ref{sec:defs}, $ab$ and $cd$ can be
assumed to be independent given $xy$ as no term in (\ref{eq:F}) mixes $ab$
and $cd$. Take the extension as in Lemma \ref{corr:E}. Using the first line
of the assumption we get
\begin{align*}
\H(w\|x)+\H(w\|y)&\le (1/k)\,\I(x,y), ~~\mbox{ and }\\
       \I(xy,a\|w)&\le (2/k)\,\I(x,y).
\end{align*}
Combining these with the second line of the assumption proves (\ref{eq:F}).
\end{proof}

\begin{theorem}\label{thm:C}
Suppose $k>1$ is an integer such that both
\begin{align*}
  \I(x,a\|y)+\I(y,a\|x) &\le (1/2^k)\,\I(x,y), \mbox{ and}\\
  \I(x,b\|y)+\I(y,b\|x) &\le (1/2^k)\,\I(x,y).
\end{align*}
Then $\ing[x,y,a,b]\ge -(3k/2^k)\,\I(x,y)$.
\end{theorem}

Observe that, under the same assumptions, Theorem \ref{thm:AB} gives the
significantly weaker estimate $\ing[x,y,a,b]\ge -(3/2^{k/2})\,\I(x,y)$.

\begin{proof}
The main ingredient is the Shannon inequality
\begin{align}\label{eq:C}
\ing[x,y,a,b] &\ge  \I(x,y\|aw)+\I(x,y\|bw)-\I(x,y\|w) - {} \\
     &~~~~~~~  - 2\big(\H(w\|x)+\H(w\|y)\big)\nonumber.
\end{align}
Let $\pi$ be the permutation that swaps $a$ and $b$ and leaves other
elements fixed. Clearly, if $f$ is an aent counterexample to the claim, then
so is $\pi f$ (as both the assumption and the conclusion are invariant for
the symmetry $\pi$). Consequently, the symmetric $(f+\pi f)/2$ is also a
counterexample. Thus, we may assume, without loss of generality, that $f$ is
symmetric for the permutation $\pi$. This symmetry extends to the base set
$xyabw_1\dots w_k$ as guaranteed by Construction \ref{lemma:Dw}.
Therefore, we have $\I(x,y\|aw_i)= \I(x,y\|bw_i)$ among other equalities. Define
$$
   \delta_i=\I(x,y\|w_i)- 2\I(x,y\|w_{i+1}).
$$
Observe that
$$
   \delta_0+2\delta_1+4\delta_2+\cdots+2^k\delta_k \le \I(x,y\|w_0)=\I(x,y);
$$
this clearly implies $\delta_i<(k/2^k)\,\I(x,y)$ for some $i\le k$.
By symmetry and by point (b) of Lemma \ref{lemma:D} we have
$$
   \delta_i = \I(x,y\|w_i)-\I(x,y\|aw_i)-\I(x,y\|bw_i).
$$
Point c) of Construction \ref{lemma:Dw} gives that for $i\le k$ we have
\begin{align*}
  \H(w_i\|x)+\H(w_i\|y) &\le i\,\big(\I(x,a\|y)+\I(y,a\|x)\big) \le {}\\
   &\le (k/2^k)\,\I(x,y).
\end{align*}
Plugging these estimates into (\ref{eq:C}), we get
$$
   \ing[x,y,a,b] \ge -\delta_i - (2k/2^k)\,\I(x,y).
$$
Since $\delta_i<(k/2^k)\,\I(x,y)$, we are done.
\end{proof}

The paper \cite{rom3} of Matveev and Romashchenko remarks that Theorems
\ref{thm:AB} and \ref{thm:C} are consequences of two infinite sequences of
entropy inequalities from \cite{M.infinf} and \cite{DFZ11}, and raises the
question whether such a sequence also exists for Theorem \ref{thm:F}. We
answer this question in the affirmative. Theorem \ref{thm:neweq} provides
such a family. Indeed, plugging in the conditions of Theorem \ref{thm:F}
into the inequality (\ref{eq:new}) gives the slightly stronger bound
$$
   \ing[x,y,a,b]+\ing[x,y,c,d]+\Big(\frac1k+\frac1k+\frac{k-1}{k^2}
   \Big)\I(x,y) \ge 0.
$$

\begin{theorem}\label{thm:neweq}
The following is an entropy inequality for any integer $k\ge 1$:
\begin{align}\label{eq:new}
  &\ing[x,y,a,b]+\ing[x,y,c,d]+\I(x,b\|y)+\I(y,b\|x) +{} \\
  & ~~~ + (1/k)\I(x,y)+(k-1)\big(\I(x,a\|y)+\I(y,a\|x)\big) \ge 0.
  \nonumber
\end{align}
\end{theorem}
\begin{proof}
Let us introduce the shorthands
\begin{align*}
  \mathcal S &= \ing[x,y,a,b]+\ing[x,y,c,d]+\I(x,b\|y)+\I(y,b\|x), \\
  \mathcal A &= \I(x,a\|y)+\I(y,a\|x), \\
  \mathcal V &= \I(x,v\|y)+\I(y,v\|x).
\end{align*}
We are going to prove the following inequality by induction on $k$:
$$
  k\,(\mathcal S + 2\mathcal V) + \I(x,y\|v)+ k(k-1)\mathcal A  \ge 0.
$$
It clearly proves (\ref{eq:new}) by setting $v$ independent of $xy$.
Using Lemma \ref{lemma:maxe} from Section \ref{sec:defs},
we may assume both $\I(ab,cd\|xy)=0$ and $\I(v,abcd\|xy)=0$. Apply the induction
hypothesis to $xv,yv,av,bv,cv,dv$  and $av$ in place of $x,y,\allowbreak
a,b,\allowbreak c,d$, and $v$.
Denote the corresponding values by $\mathcal S^*$, $\mathcal A^*$ and $\mathcal V^*$.
Observe that $\mathcal A^*=\mathcal V^*$ as both $a$ and $v$ are replaced by
$av$. The induction hypothesis provides
$$
    k\,\mathcal S^*+\I(xv,yv\|av) + k(k+1)\mathcal A^* \ge 0.
$$
The following are three Shannon inequalities:
\begin{align*}
   \mathcal S+2\mathcal V+\I(x,y\|v) & \ge \I(xv,yv\|av) -{} \\
       & ~~~~~~-5\,\I(v,abcd\|xy)-3\,\I(b,cd\|xy), \\
   \mathcal S+2\mathcal V  & \ge \mathcal S^*
      -5\,\I(v,abcd\|xy)-3\,\I(b,cd\|xy),  \\
   \mathcal A &\ge \mathcal A^* - 2\,\I(v,a\|xy).
\end{align*}
By the assumptions above, the negative terms on the right-hand sides are all zero.
Multiplying these inequalities by $1$, $k$, $k(k+1)$ times, respectively, and
adding them up, we get
\begin{align*}
  & (k+1)(\mathcal S+2\mathcal V)+\I(x,y\|v)+k(k+1)\mathcal A \ge {} \\
  & ~~~ k\mathcal S^*+\I(xv,yv\|av)+k(k+1)\mathcal A^* \ge 0,
\end{align*}
which completes the induction step.
\end{proof}


\section{Conclusion}\label{sec:conc}

In this paper, we have revisited the stability of the Ingleton inequality,
initiated in \cite{rom3}, using a traditional, tropicalization-free
framework. By completely bypassing the machinery of tropical probability
spaces, we have succeeded in streamlining the underlying concepts and
proofs, improving some of the existing estimates, and providing explicit,
non-asymptotic error terms. Additionally, we resolved an open problem posed
in \cite{rom3} by proving a new infinite family of entropy inequalities that
establishes the stability of the sum of two Ingleton expressions.

A promising avenue for future research concerns the precision of these
stability bounds. In Theorem \ref{thm:F}, the estimate depends on a square
root error term (on the order of $O(\sqrt{\epsilon})$ in terms of the
conditional mutual information $\I(x,a\|y) + \I(y,a\|x)$). In contrast,
Theorem \ref{thm:C} offers a significantly stronger, almost linear error
bound of the form $O(\epsilon \log(1/\epsilon))$.

To study whether the square root error term in Theorem \ref{thm:F} can be
tightened to an almost linear one, the following candidate family of entropy
inequalities has been verified numerically for $0\le k \le 9$:
\begin{equation}\label{eq:open-problem}
  (2^k\m-1)\mathcal S + \I(x,y\|a) +
      (k2^k\m-k)\mathcal A + k2^k\mathcal B \ge 0,
\end{equation}
where $\mathcal S = \ing[x,y,a,b] + \ing[x,y,c,d] + \mathcal A- \mathcal B$,
$\mathcal A = \I(x,a\|y) + \I(y,a\|x)$, and $\mathcal B =
\I(x,b\|y) + \I(y,b\|x)$. If inequality (\ref{eq:open-problem}) holds for
all $k \ge 1$, it would formally prove that the square root error term in
Theorem \ref{thm:F} can indeed be replaced by an almost linear one,
mirroring the relationship between Theorem \ref{thm:AB} and Theorem
\ref{thm:C}. We leave as open problems both the formal inductive proof of
(\ref{eq:open-problem}) and determining whether the proof of Theorem
\ref{thm:F} can be modified using the $a\leftrightarrow b$ symmetry \`a la
Theorem \ref{thm:C} that achieves the almost linear error bound.


\section*{Acknowledgments}
The research reported in the paper was partially supported by the ERC
Advanced grant ERMiD.



\begin{thebibliography}{99}
\bibitem{AK}
   R.~Ahlswede, P.~G\'acs, J.~K\"orner (1976),
  Bounds on conditional probabilities with applications in multi-use
  communication.
  \emph{Zeitschrift f\"ur Wahrscheinlichkeitstheorie und Verwandte Gebiete}
   {\bf 34} 157--177.
%
\bibitem{DFZ11}
  R.~Dougherty, C.~Freiling, K.~Zeger (2011),
  {Non-Shannon information inequalities in four random variables}.
  \emph{arXiv} 1104.3602.
  \url{https://doi.org/10.48550/arXiv.1104.3602}
%
\bibitem{ingleton}
 Aubrey W.~Ingleton (1971), 
 Representation of matroids.
\emph{Combinatorial mathematics and its applications}, {\bf23}, pp 149--167.
%
\bibitem{Kaced} T.~Kaced (2013)
  Equivalence of two proof techniques for non-Shannon-type inequalities.
  \emph{Proceedings of the 2013 {IEEE} ISIT}, Istanbul, Turkey, July 7-12, 236--240.
%
\bibitem{MMRV} K.~Makarychev, Yu.~Makarychev, A.~Romashchenko,
  N.~Vereshchagin (2002),
  A new class of non-Shannon-type inequalities for entropies.
  \emph{Comm. in Inf. and Systems} {\bf 2}(2) pp 147--166.
%
\bibitem{Li23}
  C. T. Li (2023),
  An Automated Theorem Proving Framework for Infor\-ma\-tion{-}Theoretic Results,
  \emph{IEEE Trans. Inf. Theory}, {\bf69}(11), pp. 6857--6877
%
\bibitem{Ma.Stud}
  F.~Mat\'u\v{s}, M.~Studen\'y (1995)
  Conditional Independences among Four Random Variables I.
  \emph{Combinatorics, Probability and Computing}, {\bf 4}(3) pp 269--278.
  \url{https://doi.org/10.1017/S0963548300001644}
%
\bibitem{mat95}
   F.~Mat\'{u}\v{s} (1995),
   Conditional independences among four random variables II.
  \emph{Combinatorics, Probability and Computing},
  {\bf4}(4) pp 407--417.
%
\bibitem{mat99}
  F.~Mat\'{u}\v{s} (1999),
  Conditional independences among four random variables III:
  Final conclusion.
  \emph{Combinatorics, Probability and Computing}, {\bf8}(3) pp 269--276.
%
\bibitem{M.infinf}
   F.~Mat\'{u}\v{s} (2007),
   Infinitely many information inequalities.
   \emph{Proceedings IEEE ISIT 2007}, Nice, France, 41--44.
%
\bibitem{M.fmadhe}
  F.~Mat\'u\v{s} (2007),
  Adhesivity of polymatroids,
  \emph{Discrete Mathematics} {\bf 307} 2464--2477.
%
\bibitem{fmtwocon}
 F.~Mat{\'u}{\v{s}} (2007),
\newblock Two Constructions on Limits of Entropy Functions.
\newblock {\em IEEE Trans. Inf. Theor.} {\bf 2007}, {\em 53},~320--330.
\newblock {\url{https://doi.org/10.1109/TIT.2006.887090}}.
%
\bibitem{entreg}
  F.~Mat\'{u}\v{s}, L.~Csirmaz (2016)
   Entropy region and convolution.
   \emph{IEEE Trans.\ Inform.\ Theory} {\bf 62} 6007--6018.
%
\bibitem{MP18}
 R.~Matveev, J.~W.~Portegies (2018),
 Asymptotic dependency structure
of multiple signals: Asymptotic equipartition property for diagrams of probability
spaces.
 \emph{Information Geometry}, {\bf1}(2) pp 237--285
%
\bibitem{MP20}
 R.~Matveev, J.~W.~Portegies (2020),
 Tropical probability theory and an application to the entropic cone. 
 \emph{Kybernetika}, {\bf56}(6) pp 1133--1153.
 \url{http://doi.org/10.14736/kyb-2020-6-1133}
%
\bibitem{MP20a}
  R.~Matveev, J.~W.~Portegies (2020),
  Tropical diagrams of probability spaces.
  \emph{Information Geometry}, {\bf3}1 pp. 61--88.
  \url{https://doi.org/10.1007/s41884-020-00027-1}
%
\bibitem{rom3}
 R.~Matveev, A.~Romashchenko (2026),
  Structural properties of entropic vectors and stability of the Ingleton
inequality.
  \emph{arXiv} 2512.02767.
  \url{https://doi.org/10.48550/arXiv.2512.02767}
%
\bibitem{oxley}
  J. Oxley (2011),
  \emph{Matroid Theory} (second edition), 
  Oxford University Press, New York (2011)
%
\bibitem{ITIP}
  N.~Rethnakar, S.~Diggavi, T.~Gläßle, E.~Perron, R.~Pulikkoonattu, 
  R.W.~Yeung, Y.~Yan (2020), 
  \emph{Online X Information Theoretic Inequalities Prover} oXitip. 
  Available at \url{https://www.oxitip.com}
%
\bibitem{studeny21}
M.~Studen\'y (2021), 
  Conditional independence structures over four discrete random variables
  revisited: conditional Ingleton inequalities.
 \emph{IEEE Transactions on Information Theory}, {\bf67}(11), 7030--7049.
 \url{https://doi.org/10.1109/TIT.2021.3104250}
%
\bibitem{W-auto}
 Shing Yin Wong, Shaocheng Liu, Linqi Song, Amin Gohari, Cheuk Ting Li
(2026),
   Automated Proving of Shannon-Type Entropy Inequalities via Fine-Tuned
   Language Models and Guided Tree Search.
   \emph{arXiv} 2606.05729.
   \url{https://doi.org/10.48550/arXiv.2606.05729}
%
\bibitem{yeung-book}
  R.~W.~Yeung (2002),
   \emph{A First Course in Information Theory}.
   Kluwer Academic/Plenum Publishers, New York.
%
\bibitem{ZhY.conic}
 Z.~Zhang, R.W.~Yeung (1997),
 A non-shannon-type conditional inequality of information quantities.
 \emph{IEEE Transactions on Information Theory}, {\bf 43}(6) 1982--1986.
%
\bibitem{ZhY.ineq}
  Z.~Zhang, R.W.~Yeung (1998),
   On characterization of entropy function via information inequalities.
   \emph{IEEE Trans.\ Inform.\ Theory} {\bf 44} 1440--1452.
%
\bibitem{ziegler}
 G.~M.~Ziegler (1994)
  \newblock{\em Lectures on polytopes}.
  \newblock Graduate Texts in Mathematics, 152 Springer.
%
\end{thebibliography}
\end{document}